\documentclass[11pt]{article}

\usepackage{tikz}
\usepackage{graphicx}
\usepackage{verbatim}
\usepackage{url}
\usepackage{fullpage}
\usepackage{amssymb,amsfonts,amsmath,amsthm}
\usepackage[colorlinks=true,linkcolor=blue,citecolor=red]{hyperref}
\usepackage[capitalise]{cleveref}

\renewcommand{\a}{\alpha}

\newcommand{\poly}{{\rm poly}}

\newcommand{\UG}{\ensuremath{\mathcal{UG}}}
\newcommand{\NP}{\ensuremath{\mathcal{NP}}}

\newcommand{\coNP}{\ensuremath{{\rm co}\mathcal{NP}}}
\newcommand{\coRP}{\ensuremath{{\rm co}\mathcal{RP}}}
\renewcommand{\P}{\ensuremath{\mathcal{P}}}
\newcommand{\N}{{\mathbb N}}

\newcommand{\eps}{\epsilon}
\newcommand{\seq}{\subseteq}

\newcommand{\MIS}{MIS}
\newcommand{\HC}{HC}
\newcommand{\MT}{{\mathsf{MT}}}

\newtheorem{theorem}{Theorem}[section]
\newtheorem{corollary}[theorem]{Corollary}

\newtheorem{definition}[theorem]{Definition}
\newtheorem{lemma}[theorem]{Lemma}
\newtheorem{claim}[theorem]{Claim}

\newtheorem{fact}[theorem]{Fact}

\newtheorem{question}[theorem]{Question}

\newenvironment{proof-sketch}{\noindent{\bf Sketch of Proof}\hspace*{1em}}{\qed\bigskip}
\newenvironment{proof-idea}{\noindent{\bf Proof Idea}\hspace*{1em}}{\qed\bigskip}
\newenvironment{proof-of-lemma}[1]{\noindent{\bf Proof of Lemma #1}\hspace*{1em}}{\qed\bigskip}
\newenvironment{proof-of-claim}[1]{\noindent{\bf Proof of Claim #1}\hspace*{1em}}{\qed\bigskip}
\newenvironment{proof-of-thm}[1]{\noindent{\bf Proof of Theorem #1}\hspace*{1em}}{\qed\bigskip}
\newenvironment{proof-attempt}{\noindent{\bf Proof Attempt.}\hspace*{1em}}{\qed\bigskip}

\newenvironment{remark}{\noindent{\it Remark.}}{\bigskip}
\newenvironment{remarks}{\noindent{\it Remarks.}}{\bigskip}

\newcommand\pnote[1]{\iftrue{\color{red} (Pasin:~#1)}\fi}

\renewcommand{\leq}{\leqslant}

\renewcommand{\geq}{\geqslant}
\renewcommand{\epsilon}{\varepsilon}

\title{Multitasking Capacity: \\Hardness Results and Improved Constructions}

\begin{document}

\author{
Noga Alon%
  \thanks{Tel Aviv University and Princeton University
  \texttt{nogaa@post.tau.ac.il}}
    \and
Jonathan D. Cohen%
  \thanks{Princeton University
  \texttt{jdc@princeton.edu}}
    \and
Thomas L. Griffiths%
  \thanks{Princeton University
  \texttt{tomg@princeton.edu}}
    \and
Pasin Manurangsi%
  \thanks{University of California, Berkeley
  \texttt{pasin@berkeley.edu}}
    \and
Daniel Reichman%
  \thanks{Princeton University
  \texttt{daniel.reichman@gmail.com}}
    \and
Igor Shinkar%
  \thanks{Simon Fraser University
  \texttt{ishinkar@sfu.ca}}
    \and
Tal Wagner%
  \thanks{MIT, CSAIL \texttt{talw@mit.edu}}
    \and
Alexander Yu%
  \thanks{Princeton University
  \texttt{ajy@princeton.edu}}
}

\maketitle

\begin{abstract}

We consider the problem of determining the maximal $\alpha \in (0,1]$
such that every matching $M$ of size $k$ (or at most $k$) in a bipartite graph $G$ contains an induced matching of size at least $\alpha |M|$. This measure was recently introduced in \cite{alon2017graph} and is motivated by connectionist models of cognition as well as modeling interference in wireless and communication networks.

We prove various hardness results for computing $\alpha$ either exactly or approximately. En route to our results, we also consider the maximum connected matching problem: determining the largest matching $N$ in a graph $G$ such that every two edges in $N$ are connected by an edge. We prove a nearly optimal $n^{1-\epsilon}$ hardness of approximation result (under randomized reductions) for connected matching in bipartite graphs (with both sides of cardinality $n$). Towards this end we define bipartite half-covers: A new combinatorial object that may be of independent interest. To the best of our knowledge, the best previous hardness result for the connected matching problem was some constant $\beta>1$.

Finally, we demonstrate the existence of bipartite graphs with $n$ vertices on each side of average degree $d$, that achieve $\alpha=1/2-\epsilon$ for matchings of size sufficiently smaller than $n/\poly(d)$. This nearly matches the trivial upper bound of $1/2$ on $\alpha$ which holds for any graph containing a path of length 3.
\end{abstract}

\section{Introduction}

A matching in an undirected graph $G$ is a set of vertex disjoint edges.
Matchings have been used in studying interference effects in parallel and distributed systems. The object of study is typically a set of units that transmit or receive information. For example, in the communication setting, there is a bipartite network $G=(S,R,E)$ consisting of senders ($S$) and receivers ($R$)\footnote{To simplify matters, we consider the synchronous setting where transmissions occur in discrete time slots.}. Given a set $E'=(s_i,t_i)_{1 \leq i \leq \ell}\subseteq E$ every sender $s_i$ wishes to send a message to its neighbor $t_i$.
The main assumption is that in order to successfully receive a message, a receiver can have only a~\emph{single} incident edge carrying a message at a given time, as messages arriving on multiple incident edges create interference with each other.
This can formalized by a condition which we term the~\emph{matching condition}: A subset $E'\subset E$ can be used for concurrent interference-free communication if it forms a matching in $G$.
However, in several communication settings such as radio and wireless networks \cite{birk1993uniform,chlamtac1985broadcasting,alon2012nearly}, a more constrained setting is considered: the senders cannot choose which edges to broadcast, but instead, if they choose to transmit, then they automatically broadcast on~\emph{all} their incident edges.
This leads to the stronger~\emph{induced matching condition}: A subset $E'\subset E$ can be used for concurrent interference-free communication if it forms an~\emph{induced} matching in $G$, namely no two edges in $E'$ are connected by an edge in $G$.

A similar interference model, directed towards understanding multitasking constraints in neural systems, has been proposed recently in computational neuroscience~\cite{feng2014multitasking,Musslick2016b,Musslick2016a}. These works seek to understand the reason behind multitasking limitations: The limited ability of people to execute several actions concurrently, a ubiquitous finding in cognitive psychology \cite{shiffrin1977controlled}. The main idea in these works is that such limitations arise from interference effects between interrelated units and not because of limited resources (e.g., limited attention or constrained working memory). The models in
~\cite{feng2014multitasking,Musslick2016b,Musslick2016a} follow the \emph{connectionist approach} to cognition \cite{rumelhart1986general,rumelhart1987parallel} which strives to explain cognition in terms of a large network of interconnected units that operate in a parallel and distributed manner and have also played a pivotal role in the study of neural networks \cite{hinton1990connectionist}. In~\cite{feng2014multitasking,Musslick2016b,Musslick2016a} a formal model to study multitasking is provided where given a bipartite graph $G=(S,R,E)$,
every vertex $s \in S$ is associated with a set of inputs $I_s$, every vertex $t\in T$ is associated with a set of outputs $O_t$ and the edge $(a,b)$ is associated with a function (``task") $f_{s,t}:I_s\rightarrow O_t$. As before, it is assumed that every unit in $S$ transmits its computed value to all adjacent neighbors in $T$, and that a value is stored without interference in a node $t \in T$ only if it receives at most one value from a single unit in $S$. In other words, given a set of $\ell$ of edges $E'=(s_i,t_i)_{1 \leq i \leq \ell}$ the set of functions $f_{s_i,t_i}$ can be performed concurrently (``multitasked") if $E'$ is an induced matching. These works are noteworthy as they relate graph theoretic properties of interconnected units and cognitive performance (multitasking). Perhaps surprisingly, hardly any works originating from connectionist models have studied how
graph-theoretic properties relate to cognitive models arising from experimental findings or computer simulations.

Based on these interference assumptions, a new measure has been proposed to capture how well such networks allow for interference-free processing~\cite{alon2017graph}. The idea behind this measure is to consider a parameter $k \leq n$, and ask whether~\emph{every} matching $M$ of size $k$ (or of size at most $k$)
contains a large~\emph{induced} matching $M'\subseteq M$.
Unless stated otherwise we will always assume that graphs are bipartite and
that both sides of the bipartition have cardinality $n$.

\begin{definition}
Let $G=(A, B,E)$ be a bipartite graph, and let $k \in \N$ be a parameter.
For $\a \in (0,1]$ we say that $G$ is a~\emph{$(k,\a)$-multitasker}
if for every matching $M$ in $G$ of size $|M| = k$,
there exists an induced matching $M' \subseteq M$ such that
\[
    |M'| \geq \a |M|.
\]
Define $\a_k(G)$ to be the maximal $\a$ such that $G$ is a $(k,\a)$-multitasker if $G$ contains a matching of size $k$,
and define $\a_k(G) = 1$ if $G$ does not contain a matching of size $k$.
We call the parameter $\a_k \in (0,1]$ the \emph{multitasking capacity of $G$ for matchings of size $k$}.

Also, define $\a_{\leq k}(G) = \min_{1 \leq \ell \leq k} \a_\ell(G)$
and call it the~\emph{multitasking capacity of $G$ for matchings of size at most $k$}.\footnote{Since we consider the minimum, the definition of $\a_{k}$ ensures that values of $r \leq k$ for which there is no matching of size $r$
have no influence on $\a_{\leq k}(G)$.}
\end{definition}
The parameters $\a_k, \a_{\leq k}$ measure how resilient $G$ is to interferences. The larger these parameters are,
the better $G$ is considered as a multitasker.

\medskip

One motivation for this definition is that it is sometimes assumed (e.g., \cite{feng2014multitasking}) that the set of tasks (edges) that need to be multitasked are restricted to be a matching, a restriction which is imposed by limitation on the parallelism of the network. In this case, the multitasking capacity quantifies for~\emph{every} set $S'$ of allowable tasks what fraction of tasks from $S'$ are guaranteed to be achieved without interference. Another motivation is the distinction between interference effects that result from a violation of the matching condition to those that result from a violation of the \emph{induced} matching condition. In particular, the above multitasking measure allows us to assess the fraction of tasks that can be performed concurrently conditioned on not violating the matching condition.

In~\cite{alon2017graph} several properties of $\a_{\leq n}(G)$ have been proven. For example, it was shown that $\a_{\leq n}(G)\leq \frac{9}{\sqrt{d}}$ for $d$-regular graphs, and that $\a_{\leq n}(G)\leq O((\frac{\log n}{d})^{1/3})$ for graphs of average degree $d$.
This upper bound supports the previous hypothesis~\cite{feng2014multitasking} suggesting that as the average degree increases, the multitasking capacity inevitably decreases, regardless of the structure of the network -- a phenomenon referred to as the ``multiplexing versus multitasking trade-off''%
\footnote{In the irregular case this holds assuming the average degree satisfies $d \gg \log n$.}.
It was also shown in~\cite{alon2017graph} how to construct graphs with desirable multitasking properties. Namely graphs for which $\a_{\leq k}(G) \geq \tau$ for $\tau=\Omega(1)$ provided that $k=O(n/d^{1+\tau})$, where $d$ is the average degree of $G$.

The results in~\cite{alon2017graph} leave several questions.

\begin{question}\label{q:compute-alpha}
    Given a graph $G$ and a parameter $k$, can we compute $\a_k(G)$ or $\a_{\leq k}(G)$ efficiently?
\end{question}
Indeed, if we are to use $\a_k(G)$ or $\a_{\leq k}(G)$ to evaluate how prone to interference parallel architectures are, then a natural question is whether it is possible to compute or approximate these quantities in polynomial time. For example, computer simulations are frequently used in developing connectionless models and these models often consist of networks consisting of dozens (or more) of units. Hence to evaluate the usefulness of $\a_k(G)$ in connectionist models of multitasking it is desirable to have efficient methods to compute $\a_k(G)$ exactly or approximately.

Another question is whether it is possible construct multitaskers with near-optimal capacity.
While~\cite{alon2017graph} provide multitaskers with $\a_{\leq k}(G)=\Omega(1)$ for $k\leq n/d^{O(1)}$ (and show that the upper bound on $k$ is tight up to the degree of the $d^{O(1)}$ term), the best constant value of $\a_{\leq k}(G)$ they achieve is bounded away from the natural barrier\footnote{Observe that if a network contains a path of length $3$ then trivially $\a_{\leq k}(G) \leq 1/2$ for all $k \geq 3$.} $\a_{\leq k}(G)\leq1/2$.
We thus raise the following question.
\begin{question}\label{q:construct-alpha-0.5}
    Is there an infinite family of graphs $G_n$ of average degree $d$ such that $\a_{\leq k}(G_n) \geq 1/2-\eps$
    for arbitrarily small $\eps > 0$ and $k \geq n/d^{f(\epsilon)}$ for some function $f$?
\end{question}

In this paper we address these two questions.
For \cref{q:compute-alpha} we show that under standard complexity theoretic assumptions $\a_k(G)$ and $\a_{\leq k}(G)$ cannot be computed efficiently, thus giving a negative answer to this question.
For \cref{q:construct-alpha-0.5} we give a positive answer, by showing how to construct bipartite graphs with multitasking capacity approaching $1/2$.

\subsection{Our results}
As it turns out, a useful notion in studying the computational hardness of computing the multitasking capacity is that of a~\emph{connected matching}, which is a matching in which every two edges are connected by a third edge (see Definition~\ref{def:connected_matching} for a formal definition). Connected matchings have been studied in several contexts, such as Hadwiger's conjecture~\cite{kawarabayashi2005improvements,plummer2003special,furedi2005connected}. Motivated by applications to other optimization problems~\cite{jobson2014connected}, algorithms for finding connected matchings of maximum cardinality have been studied in special families of graphs such as chordal~\cite{cameron2003connected} and bipartite chordal graphs~\cite{jobson2014connected}\footnote{Observe that bipartite chordal graphs are not necessarily chordal. See~\cite{jobson2014connected} for details.} and bipartite permutation graphs~\cite{golovach2014hadwiger}.

In \cref{sec:connected} we establish hardness of approximation for the size of the largest connected matching to within a factor of $n^{1-\epsilon}$ assuming $\NP\neq\coRP$. Previously, this problem was known to be $\NP$-hard to approximate within some constant factor~\cite{plummer2003special} for general (non-bipartite) graphs. We also prove that deciding whether a bipartite graph $G=(A,B,E)$ with $|A|=|B|=n$
contains a connected matching of size $n$ is $\NP$-hard.

In \cref{sec:hardness} we prove several hardness results for computing the multitasking capacity. To be more precise, we define the decision problem of computing the multitasking capacity as follows: 
\begin{definition} \label{def:mt}
Let $\MT$ be the problem of deciding whether for a given graph $G$, a positive integer $k \in \N$ and a rational number $\eta>0$ it holds that $\alpha_{k}(G) \geq \eta$.
\end{definition}
The problem $\MT$ belongs to the second level of the polynomial hierarchy, $\Pi_2$,  since the statement $\alpha_{k}(G) \geq \eta$ can be expressed as $\forall M \exists M'  P(G,k; M,M')$, where $P$ is the predicate checking that $M$ is a matching in $G$ of size $k$, and $M'$ is an induced matching, which is clearly computable in time $\poly(|G|)$.
We note that it is not clear whether it belongs to $\NP$ or to $\coNP$, and in fact, we give evidence that $\MT$ belongs to neither of the classes.
Specifically, we show that $\MT$ is both $\NP$-hard and $\coNP$-hard; thus, if $\MT \in \NP \cup \coNP$, then the polynomial hierarchy collapses to the first level.

Furthermore, we show various hardness of approximation results for computing $\alpha_k(G)$ and $\alpha_{\leq k}(G)$.
Most notably, we show under standard complexity theoretic assumptions that
(1) $\alpha_n(G)$ is inapproximable to within $n^{1-\epsilon}$ for any $\epsilon>0$, and, (2) $\alpha_{\leq k}(G)$ inapproximable to within any constant for $k=n^{1-\eps}$ for any $\eps>0$. Furthermore, under a stronger assumption, we improve the inapproximability ratio for $\alpha_{\leq k}(G)$ to $n^{1/\mathrm{polyloglog}(n)}$ for $k=n^{1-1/\mathrm{polyloglog}(n)}$.
Our hardness results are summarized in \cref{tbl:hardness}.

\begin{table}
 \begin{tabular}{||c c c c c ||}
 \hline
 Variant & Assumption & $k$ & $f$ (approximation factor) & Remarks \\ [0.5ex]
 \hline\hline
 $\alpha_k(G)$ & $\P\neq\NP$ & $n$ & $n^{1-\epsilon}$ for any $\epsilon>0$ &   \\
 \hline
 $\alpha_k(G)$ & $\P\neq\NP$ & $n$ & $O(d/\poly\log d)$ & $G$ has maximum degree $d$ \\
 \hline
 $\alpha_{\leq k}(G)$ & $\NP\neq\coRP$ & $n$ & some constant & \\
 \hline
 $\alpha_{\leq k}(G)$ & $\NP\neq\coRP$ & $n^{1-\epsilon}$ & arbitrarily large constant & \\
 \hline
 $\alpha_{\leq k}(G)$ & ETH & $n^{1-1/\mathrm{polyloglog}(n)}$ & $n^{1/\mathrm{polyloglog}(n)}$ & \\
 \hline
\end{tabular}
\caption{Hardness of approximation results for computing the multitasking capacity. In each row, the stated variant of the multitasking capacity (either $\alpha_k(G)$ or $\alpha_{\leq k}(G)$) is hard to approximate under the stated assumption up to a multiplicative factor $f$, for the stated values of $k$ and $f$.}
\label{tbl:hardness}
\end{table}

In \cref{sec:construction}, we prove the existence of multitaskers with near-optimal capacity. For integers $d,n$ with $n \geq d$ and $\epsilon \in (0,1)$, we show how to construct multitasker graph $G$ on $2n$ vertices with average degree $d$ and $\alpha_{\leq k}(G) \geq 1/2-\epsilon$, where $k = \Omega(n/d^{1+O(1/\epsilon)})$. In particular, for $d=n^{o(1)}$ this implies that $\epsilon$ can be taken to be $o(1)$, and thus $\alpha_{\leq k}(G)$ tends to its natural barrier $1/2$ as $n$ grows.

\subsection{Our techniques}

\paragraph{Hardness results.}
With respect to multitasking, connected matchings are the worst possible multitasking configuration for a matching of size $k$. In particular, it holds trivially that $\a_k(G)\geq1/k$ and $\a_{\leq k}(G) \geq 1/k$, and the equality holds if and only if $G$ contains a connected matching of size $k$. This fact, together with extremal Ramsey bound on the size of independent sets, turns out to be instrumental in proving hardness results for computing the multitasking capacity.

\paragraph{Construction of multitaskers.}
The starting point of our multitaskers with nearly optimal multitasking capacity is based on locally sparse graphs, similarly to~\cite{alon2017graph}.
They used the local sparsity with Turan's lower bound on independent sets in graphs with a given average degree in order to establish the existence of sufficiently large independent sets (which translate to induced matchings). However, the use of Turan's bound necessarily entails a constant loss, which makes the final multitasking capacity bounded away from $1/2$. We circumvent this roadblock by also requiring that the graph has large girth, and use this fact in order to carefully construct a large independent set.

\section{Preliminaries}\label{sec:prelim}
All graphs considered in this work are undirected.
A matching in a graph $G=(V,E)$ is a collection $M \seq E$ of vertex disjoint edges.
We say that a vertex $v \in V$ is \emph{covered} by $M$ if it is one of the endpoints of an edge in $M$.
We say that a matching $M$ is \emph{induced in $G$} if no two edges in $M$ are connected by an edge in $E$,
i.e., the vertices in $M$ span only the edges in $M$ and no other edges.
Given a graph $G$ and an edge $e = (u,v) \in E$, we define the \emph{contraction} of $e$ to be the operation that
produced the graph $G \setminus e$, whose vertex set is $(V\cup v_{e}) \setminus \{u,v\}$, the vertex $v_{e}$ is connected to all vertices in $G$ neighboring $u$ or $v$,
and for all other vertices $x,y \in V \setminus\{u,v\}$, they form an edge in $G \setminus e$ if and only if they were connected in $G$.
Contracting a set of edges, and in particular contracting a matching, means contracting the edges one by one in an arbitrary order%
\footnote{We remark that the graph obtained from contracting a set of edges, indeed, does not depend on the order.}.

Below we define two combinatorial optimization problems that we will relate to when proving hardness of approximation results for the parameters $\a_k$ and $\a_{\leq k}$.

\begin{definition}\label{def:mis}
Given an undirected graph $G$, an \emph{independent set} in $G$ is a set of vertices that spans no edges.
The Maximum Independent Set Problem ($\MIS$) is the problem of finding a maximum cardinality of an independent set in $G$.
\end{definition}

\begin{definition}\label{def:biclique}
Given a graph $G = (V,E)$, we say that two disjoint subsets of the vertices $A,B \seq V$ form a \emph{bipartite clique (biclique)} in $G$ if $(a,b) \in E$ for all $a \in A$ and $b \in B$.
We say that the biclique $(A,B)$ is balanced if $|A| = |B|$.
In the Maximum Balanced Biclique Problem we are given a bipartite graph $G$ and a parameter $k$, and the goal is to decide whether $G$ contains a balanced biclique with $k$ vertices on each size.
\end{definition}

\begin{definition}\label{def:connected_matching}
Given a graph $G$, a \emph{connected matching} in $G$ is a matching $M$ such that every two edges in $M$ are connected by an edge in $G$.
We use $\nu_c(G)$ to denote the size of the maximum cardinality of a connected matching in $G$.
In the Connected Matching Problem, we are given graph $G$ and parameter $k$ and our goal is to determine whether $\nu_c(G) \geq k$.
\end{definition}


Given an optimization (minimization or maximization) problem $\Pi$ over graphs, we denote by $OPT_\Pi(G) > 0$ the value of the optimal solution of $\Pi$ for $G$.
An algorithm $A$ for a maximization (minimization) problem is said to achieve an approximation ratio $\rho > 1$ if for every input $G$
the algorithm returns a solution $A(G)$ such that $OPT_\Pi(G)  \geq A(G) \geq OPT_\Pi(G)/\rho$ (resp.\ $OPT_\Pi(G) \leq A(G) \leq \rho \cdot OPT_\Pi(G)$).

We assume familiarity with complexity classes such as $\NP, \coNP, \coRP, \Pi_2$, and the polynomial-time hierarchy. Precise definitions of these terms are omitted, and can be found, e.g., in \cite{papadimitriou2003computational}.

\section{Hardness results for maximum connected matchings}\label{sec:connected}

In this section, we prove hardness results for finding large connected matchings in graphs.

\subsection{Hardness of approximating the size of a maximum connected matching}

We start by showing an almost optimal hardness of approximation result for the connected matching problem.

\begin{theorem}\label{thm:connected-match-inapprox}
Given a bipartite graph $G$ with $n$ vertices on each side,
it is $\NP$-hard to approximate $\nu_c(G)$ within a factor of $n^{1-\eps}$ for any $\eps > 0$
under a randomized polynomial time reduction.

More precisely, given a bipartite graph $G$ with $n$ vertices on each side,
it is $\NP$-hard to distinguish between the case where $\nu_c(G) \geq n^{1-\eps}$
and the case where $\nu_c(G) \leq n^{\eps}$ for any $\eps > 0$.
\end{theorem}

A natural approach to prove hardness of approximation results for connected matching is to reduce the clique problem to it. Namely given a graph $G = ([n], E_G)$ for which we wish to determine if $G$ contains a $k$-clique, replace every vertex $i$ by an edge $e_i=(u_i, v_i)$ and add two edges $(u_i, v_j)$ and $(u_j, v_i)$ for every edge $(i, j)$ in $G$. Call the resulting graph after these transformation $G'$. While it is clear that a large clique in $G$ translates to a large connected matching in $G'$, it is not clear that a large connected matching in $G'$ implies a large clique in $G$. The difficulty is that a connected matching might contain ``bad" edges of the form $(u_i,v_j)$ where $i \ne j$. An illustrative example is the case where $G = K_{n/2, n/2}$ is a biclique; in this case, the largest clique in $G$ has size only $2$ but the resulting graph $G'$ contains a large connected matching of size as large as $n$.

To overcome this problem, we first observe that instead of adding both $(u_i, v_j)$ and $(u_j, v_i)$ to the graph $G'$ for every edge $(i, j)$ in $G$. It suffices to add only one of the two to retain a large connected matching in the YES case. Then, the insight is that, when we choose the edge to add independently at random for each $(i, j)$, we can control the number of bad edges in every connected matching in $G'$


We formalize the described ideas below, starting with the main gadget of our reduction:

\begin{definition}
Fix $n \in \N$. A bipartite graph $\HC_n = (A = \{u_1, \dots, u_n\}, B = \{v_1, \dots, v_n\}, E_H)$ is said to be a \emph{bipartite half-cover of $K_n$}
if (1) for every $\{i, j\} \subseteq [n]$, $(u_i, v_j) \in E_H$ or $(u_j, v_i) \in E_H$, and (2) for every $i \in [n]$, $(u_i, v_i) \notin E_H$.
\end{definition}

The reduction used in the proof of \cref{thm:connected-match-inapprox} uses the existence of such bipartite half-covers of $K_n$ that do not contain a large connected matching.
Such graphs can be easily constructed using a randomized algorithm as shown below.

\begin{claim} \label{claim:random-half-cover-small-nu-c}
There is an $O(n)$-time randomized algorithm that on input $n \in \N$ outputs a graph $HC_n$, which is a bipartite half-cover of $K_n$ such that $\nu_c(HC_n) \leq O(\log n)$ with probability $1 - o(1)$.
\end{claim}

\begin{proof}
We construct $\HC_n$ by choosing for each $\{i, j\} \subseteq [n]$ to add to $E_H$ either $(u_i, v_j)$ or $(u_j, v_i)$ independently with probability 1/2.
Clearly, $\HC_n$ is a bipartite half-cover of $K_n$.
Below we show that $\nu_c(H) \leq O(\log n)$ with probability $1 - o(1)$.
We prove this in two steps: first, we will prove the $O(\log n)$ upper bound on a special class of connected matching and, then, we will show that any connected matching contains a large (constant fraction) matching of this type.

Let $M \subseteq E_H$ be any matching in $H$. We say that the matching is \emph{non-repetitive} if, for each $i \in [n]$, at most one of $u_i$ or $v_i$ appears in $M$. We will now argue that with probability $1 - o(1)$, any connected non-repetitive matching has size less than $D := 20\log n$. To do so, consider any ordered tuple $(i_1, j_1, \dots, i_D, j_D)$ where $i_1, \dots, i_D, j_1, \dots, j_D$ are all distinct. The probability that $(u_{i_1}, v_{j_1}), \dots, (u_{i_D}, v_{j_D})$ is a connected matching is at most
\begin{align*}
\Pr[\forall 1 \leq k < \ell \leq D, (u_{i_k}, v_{j_\ell}) \in E_H \vee (u_{i_\ell}, v_{j_k}) \in E_H] &= \prod_{1 \leq k < \ell \leq D} \Pr[(u_{i_k}, v_{j_\ell}) \in E_H \vee (u_{i_\ell}, v_{j_k}) \in E_H] \\
&= \prod_{1 \leq k < \ell \leq D} (3/4) = (3/4)^{D(D - 1)/2}
\end{align*}
where the first two equalities use the fact that $i_1, \dots, i_1, j_1, \dots, j_D$ are distinct, meaning that the events considered are all independent. Hence, by union bound over all such sequences, we can conclude that the probability that $H$ contains a connected non-repetitive matching of size $D$ is at most $n^{2D} \cdot (3/4)^{D(D - 1)/2} = (n^2 \cdot (3/4)^{(D - 1)/2})^D = o(1)$.

Finally, observe that any matching $M \seq E_H$ contains a non-repetitive matching $M' \seq M$ of size at least $|M|/3$.
Indeed, given a matching $M$ we can construct $M'$ iteratively by picking an arbitrary edge $e = (u_i, v_j) \in M$,
remove $e$ and all edges touching $v_i$ or $u_j$ from $M$ and add $e$ to $M'$.
We repeat this procedure until $M = \emptyset$.
Since we add one edge to $M'$ while removing at most three edges from $M$, we arrive at a non-repetitive $M' \subseteq M$ of size at least $|M|/3$.
As a result, the graph $\HC_n$ does not contain any connected matching of size at least $3D = O(\log n)$ with probability $1 - o(1)$.
\end{proof}

\begin{remarks}
\begin{enumerate}
\item
We remark that a deterministic polynomial time construction of such graphs would imply that the hardness
result in \cref{thm:connected-match-inapprox} holds under a deterministic reduction (as oppose to the randomized reduction, currently stated).
\item
We comment that there is a connection between Ramsey graphs and half-cover of $K_n$ with small $\nu_c(\HC_n)$. Specifically, if we can deterministically construct half-cover for $K_n$ with $\nu_c(\HC_n) \leq f(n)$, then we can deterministically construct $n$-vertex $(f(n)+1)$-Ramsey graphs. This is because, we can think of half-cover $\HC_n$ as a bichromatic $K_n$ where $(i, j)$ for $i < j$ is colored red if $(u_i, v_j) \in E_H$ and it is colored blue otherwise (i.e. $(u_j, v_i) \in E_H$). It is easy to check that any monochromatic clique of size $r$ in $K_n$ implies a connected matching of size $r-1$ in $\HC_n$.
While there are explicit constructions of Ramsey graphs, it is unclear (to us) how to construct such half-cover from these constructions.
\item
Using a different approach we can show that it is $\NP$-hard to compute $\nu_c(G)$ under a \emph{deterministic} reduction.
See \cref{sec:hardness-deterministic-reduction-cm} for details.
\end{enumerate}
%
\end{remarks}

\subsubsection{Proof of \cref{thm:connected-match-inapprox}}

With the gadget from \cref{claim:random-half-cover-small-nu-c} we are ready to prove \cref{thm:connected-match-inapprox}.
This is done in the following claim.

\begin{claim} \label{claim:reduction-cm}
Let $G = (V_G = [n], E_G)$ be an $n$-vertex graph, and let $H = (A = \{u_1, \dots, u_n\}, B = \{v_1, \dots, v_n\}, E_H)$ be a balanced bipartite graph.
Let $G \boxminus H = (A, B, E_{G \boxminus H})$ be the balanced bipartite graph with $n$ vertices on each side,
where (1) for every $\{i, j\} \subseteq [n]$, $(u_i, v_j) \in E_{G \boxminus H}$ if and only if $(u_i, v_j) \in E_H$ and $(i, j) \in E_G$, and
(2) for every $i \in [n]$, $(u_i, v_i) \in E_{G \boxminus H}$.

Then, for any such $G$ we have $\nu_c(G \boxminus H) \leq \omega(G) + 3\nu_c(H)$ where $\omega(G)$ denotes the clique number of $G$.
Furthermore, if $H$ is a bipartite half-cover of $K_n$, then $\omega(G) \leq \nu_c(G \boxminus H)$.
\end{claim}

\cref{claim:reduction-cm} immediately implies \cref{thm:connected-match-inapprox}.
Indeed, by \cite{haastad2001some,zuckerman2006linear} given an $n$-vertex graph $G$ it is NP-hard to decide between the case where $\omega(G) \geq n^{1-\eps/2}$, and the case where $\omega(G) \leq  n^{\eps/2}$.
Therefore, we can define a randomized reduction that given an $n$-vertex graph $G$ constructs (with high probability) $\HC_n$, the bipartite half-cover of $K_n$, with $\nu_c(\HC_n) \leq O(\log n)$,
and outputs $G \boxminus H$, which can be clearly constructed in time that is linear in the size of $G$.
In the YES case, if $\omega(G) \geq n^{1-\eps/2}$, then by the ``furthermore'' part of \cref{claim:reduction-cm} we have $\nu_c(G \boxminus \HC_n) \geq \omega(G) \geq n^{1-\eps/2}$,
and in the NO case, if $\omega(G) \leq n^{\eps/2}$, then by \cref{claim:reduction-cm} we have $\nu_c(G \boxminus \HC_n) \leq \omega(G) + \nu_c(\HC_n) \leq n^{\eps/2} + O(\log n)$.
This completes the proof of \cref{thm:connected-match-inapprox}.

We now turn to the proof of \cref{claim:reduction-cm}.

\begin{proof}[Proof of \cref{claim:reduction-cm}]
First, we will show that $\nu_c(G \boxminus H) \leq \omega(G) + 3\nu_c(H)$. Let $M \subseteq E_{G \boxminus H}$ be any connected matching in $G \boxminus H$. We partition $M$ into two disjoint sets $M_{\parallel}$ and $M_{\times}$ where $M_{\parallel} = M \cap \{(u_i, v_i) \mid i \in [n]\}$ and $M_{\times} = M \setminus M_{\parallel}$. We will show that $|M_{\parallel}| \leq \omega(G)$ and $|M_{\times}| \leq 3\nu_c(H)$.

To show that $|M_{\parallel}| \leq \omega(G)$, suppose that $M_{\parallel} = \{(u_{i_1}, v_{i_1}), \dots, (u_{i_t}, v_{i_t})\}$. By the definition if $(u_i, v_i)$ is connected to $(u_{i'}, v_{i'})$ in $G \boxminus H$, then $(i, i') \in E_G$.
Therefore, $\{i_1, \dots, i_t\}$ induces a clique in $G$ and $\omega(G) \geq t = |M_{\parallel}|$ follows.

Next, we show that $|M_{\times}| \leq 3\nu_c(H)$. Let us first define non-repetitive matching in the same way as that in the proof of \cref{claim:random-half-cover-small-nu-c}.
Using the same argument as in that proof, we can conclude that $M_{\times}$ contains a non-repetitive connected matching $M'_{\times} \subseteq M_{\times}$ of size at least $|M_{\times}|/3$. We claim that $M'_{\times}$ is also a connected matching in $H$.
Indeed, since every edge in $M'_{\times}$ belongs to $E_H$, the non-repetitiveness implies that any pair of edges in $M'_{\times}$ is  connected by an edge that also belongs to $E_H$.
As a result, we can conclude that $|M_{\times}| \leq 3|M'_{\times}| \leq 3\nu_c(H)$.

Combining the above two bounds yields $\nu_c(G \boxminus H) \leq \omega(G) + 3\nu_c(H)$ as desired.

\medskip

Finally, assume that $H$ is a bipartite half-cover of $K_n$. For any clique $C \subseteq V_G$ in $G$, it is not hard to see that the matching $M_C = \{(u_i, v_i): i \in C \}$ is a connected matching in $G \boxminus H$.
Indeed, for each distinct $i, j \in C$ we have either $(u_i, v_j) \in E_H$ or $(u_j, v_i) \in E_H$ (from definition of bipartite half-cover of $K_n$),
and hence either $(u_i, v_j)$ or $(u_j, v_i)$ belongs to $E_{G \boxminus H}$.
Therefore, $\nu_c(G \boxminus H) \geq \omega(G)$, which completes our proof.
\end{proof}

\subsection{Hardness of finding a connected perfect matching}

In this section we show that given a bipartite graph $G$ with $n$ vertices on each side,
it is $\NP$-hard to find a connected matching of size $n$.

\begin{theorem}\label{thm:hardness-perfect-connected-matching}
Given a bipartite graph $G=(A,B,E)$ with $|A|=|B|=n$ it is $\NP$-hard to determine whether $\nu_C(G)=n$.
\end{theorem}

\begin{proof}
By \cref{thm:connected-match-inapprox} given a graph $G = (A,B,E_G)$ with $N$ vertices of each side it is $\NP$-hard to decide whether $G$ contains a connected matching of size $k = N^{1-\eps}$.
Consider the reduction that given a graph $G = (A,B,E_G)$ outputs $H=(A \cup A', B \cup B',E_H)$ as follows.
The sets $A'$ and $B'$ are two disjoint sets that are also disjoint from $A,B$ with $|A'|=|B'|=N-k$.
The set of edges $E_H$ is defined as $E_H = E_G \cup \{(i,j) : i \in A', j \in B \cup B'\} \cup \{(i,j) : i \in A \cup A', j \in B\} $.
That is, the graph $H$ contains the graph $G$ as the induced graph on the vertices $A \cup B$,
and in addition, every vertex in $A'$ is connected to all vertices in $B \cup B'$,
and every vertex in $B'$ is connected to all vertices in $A \cup A'$,

The graph $H$ is a balanced bi-partite graph with $n = 2N-k$ vertices on each side.
We claim that $\nu_C(G)=k$ if and only if $\nu_C(H)=n$.

In one direction, suppose that $G$ has a connected matching $M_G = \{e_1,...,e_k\}$ of size $k$. We construct a matching $M'$ of size $2N-k$ as follows.
For each vertex $v \in A \cup B$ not covered by $M_G$, we pick a distinct element $w_v \in A' \cup B'$ that is a neighbor of $v$.
Define a matching in $H$ to be $M' = M \cup N$, where $N = \{(v,w_v) : v \in V(G) \setminus V(M_G)\}$.
By the construction of $H$, each edge in $N$ is connected to every other edge in $M'$ using an edge between $A'$ and $B'$.
Every pair of edges in $M_G$ are connected since $M_G$ is a connected matching in $G$.
Thus, $M'$ is a connected matching of size $n$ in $H$.

Conversely, suppose $H$ has a connected matching $M_H$ of size $n$. Then, there must be is a submatching $M \subseteq M_H$ of size $|M| = k$ such that no edge in $M$ contains a vertex in $A'\cup B'$.
Thus, $M$ is a matching in $G$, and since $M_H$ is a connected matching so is $M$. It follows that $G$ has a connected matching of size $k$, as required.
\end{proof}

\section{Hardness results for computing $\a_k(G)$}\label{sec:hardness}

In this section we study the computational complexity both of the decision problem $\MT$
as well as the problem of computing $\a_k(G)$ exactly or approximately.
We first show an almost optimal inapproximability result for $\alpha_n(G)$, which is stated and proved below.

\begin{theorem}\label{thm:alpha-n-hardness}
    For any $\eps>0$, given a bipartite graph $G$ with $n$ vertices in each part,
    it is $\NP$-hard to approximate $\a_n(G)$ within a factor $n^{1-\eps}$.

    Furthermore, given a bipartite graph $G$ with $n$ vertices in each part,
    where the degree of each vertex is at most $d$
    it is $\NP$-hard to approximate $\a_n(G)$ within a factor $O(\frac{d}{\log^4(d)})$ 
    and it is $\UG$-hard to approximate $\a_n(G)$ within a factor $O(\frac{d}{\log^2(d)})$.
\end{theorem}

\begin{proof}
  The proof is by a reduction from the Maxium Independent Set problem.
  Given an $n$ vertex graph $H = (U_H,E_H)$ instance of the $\MIS$ we construct
  a bipartite graph $G$ as follows. Denote the vertices of $H$ by $U_H = \{u_1,u_2,\dots,u_n\}$.
  Then the vertices of the bipartite graph $G = (V_G = A \cup B, E_G)$
  are defined by $A = \{v_i : i \in [n]\}$ and $B = \{v'_i : i \in [n]\}$,
  and the edges of $G$ are $E_G = \{(v_i,v'_i) : i \in [n]\} \cup \{(v_i,v'_j) :  i < j \wedge (u_i,u_j) \in E_H\}$.
  Note that the only perfect matching in $G$, i.e., a matching of size $n$,
  is the matching $N = \{(v_i,v'_i) : i \in [n]\}$. Indeed, suppose there exists another matching $M$ with $|M| = n$. Then $M$ has at least one edge of the form $e=(v_i,v'_j)$ with $i<j$ and suppose that $e$ is such that $i$ is minimal (where the minimum is taken with respect to all edges not in $N$). If any edge in $M$ covers $v'_i$, then it cannot belong to $N$ as $M$ is a a matching. By the definition of $E_G$ there cannot be an edge in $M$ that covers $v'_i$ by the minimality of $i$. As all vertices of $H$ must be matched in order for $|M| = n$, we get a contradiction showing that $N$ is indeed the unique matching of size $n$.

  We claim that $H$ contains an independent set of size at least $\a$ if and only if $\a_n(G) \geq \frac{\a}{n}$.
  Indeed, a set $I \seq V_H$ is an independent set in $H$ if and only if
  $M' = \{(v_i,v'_i) : i \in I\}$ is an induced matching contained in $M$. Hence if $H$ contains an independent set of size
  $\alpha$ then $M$ contains an induced matching of size $\alpha$. Conversely, If $M$ contains an induced matching of size $\alpha$ then
  $H$ has an independent set of size $\alpha$.
  It is well known that for any $\delta<1/2$ it is \NP-hard to distinguish between $n$-vertex graphs that contain an independent set of size
  at least $n^{1-\delta}$ (YES case) and graph that do not contain an independent set of size at least $n^{\delta}$ (NO-case) \cite{haastad2001some,zuckerman2006linear}.
  By the reduction described above it is \NP-hard to distinguish between a bipartite graph $G'$ with sides of cardinality $n$ satisfying
  $\a_n(G') \geq n^{1-\delta}/n=n^{-\delta}$ to a graph $G''$ satisfying $\a_n(G') \leq n^{\delta}/n=n^{\delta-1}$ as this would enable to distinguish between
  the YES and NO cases described above. The result now follows by taking $\delta$ to equal $\epsilon/2$.

  The result for graphs of maximum degree $d$ follows by noting that if the
  maximal degree of $H$ is at most $d$, then the maximal degree of $G$ is upper bounded by $d+1$.
  Therefore, since it is \NP-hard to approximate $\MIS$
  in graphs of maximum degree $d$ within a factor of $O(\frac{d}{\log^4(d)})$ \cite{chan2016approximation} and \UG-hard to approximate $\MIS$ in graphs of maximum degree $d$ within  a factor of $O(\frac{d}{\log^2(d)})$ \cite{austrin2009inapproximability},
  the analogous hardness computing $\a_n$ also follows.
\end{proof}

We remark that by adding isolated vertices to the graph, the above hardness result also implies hardness of approximating $\alpha_{k}(G)$ to within factor of $k^{1 - \varepsilon}$ for every $\varepsilon > 0$ and every $k \geq n^\delta$ for any constant $\delta \in (0, 1)$.

Recall the decision problem $\MT$ from Definition~\ref{def:mt}. As mentioned in the introduction, $\MT$ clearly belongs to the class $\Pi_2$. We show the following:

\begin{theorem}\label{thm:MT-is-NP-hard-coNP-hard}
  The decision problem $\MT$ is $\NP$-hard and $\coNP$-hard.
\end{theorem}


\begin{proof}[Proof of \cref{thm:MT-is-NP-hard-coNP-hard}]
    By \cref{thm:alpha-n-hardness} if follows that that
    there is a reduction from any problem in $\NP$ that produces a graph $G$ and a parameter $k = n$
    such that in the YES case $\a_k(G) \geq 1/n^{\eps}$, and in the NO case $\a_k(G) \leq 1/n^{1-\eps}$.
    In particular, this implies that $\MT$ is $\NP$-hard.

    In order to prove that $\MT$ is $\coNP$-hard we use \cref{thm:hardness-perfect-connected-matching}.
    Indeed, observe that $\a_{n}(G) \leq 1/n$ if and only if $G$ contains a connected matching of size $n$,
    and hence there is a reduction from any problem in $\NP$ that produces a graph $G$ and $k = n$
    such that in the YES case $\a_k(G) \leq 1/k$, and in the NO case $\a_k(G) \geq 2/k$.
    This completes the proof of \cref{thm:MT-is-NP-hard-coNP-hard}
\end{proof}

Using Theorem~\ref{thm:hardness-perfect-connected-matching}, we demonstrate that it is unlikely that $\MT$ belongs to $\NP \cup \coNP$.
\begin{corollary}
If the decision problem $\MT$ belongs to $\NP \cup \coNP$, then the polynomial-time hierarchy collapses to the first level.
\end{corollary}
Indeed, this follows from the fact that if $\NP \seq \coNP$, then $\NP =\coNP$ (see e.g., Proposition 10.2 in~\cite{papadimitriou2003computational}), and hence the polynomial hierarchy collapses to the first level.

We end this section with several remarks.
    \begin{enumerate}
    \item
    Note that the proof of \cref{thm:alpha-n-hardness} shows that the problem of computing $\a_n(G)$
    is $\NP$-hard on graphs with $n$ vertices on each side that
    even if $G$ contains a unique perfect matching.
    \item
    Note also that the hardness result in \cref{thm:alpha-n-hardness} for bounded degree graphs is unlikely to hold for $d$ \emph{regular} graphs (as opposed to graphs with degree at most $d$)
    This is because in \cite{alon2017graph} it is shown that $\a_n(G) \leq O(1/\sqrt{d})$ for every $d$-regular graph $G$.
    In particular, this implies that it is easy to approximate $\a_n(G)$ within a factor of $O(\sqrt{d})$ for $d$-regular graphs.
    \end{enumerate}


\section{Hardness results for computing $\a_{\leq k}(G)$}

Here we prove that it is hard to calculate the parameter $\a_{\leq k}(G)$.

\subsection{Hardness results for computing $\a_{\leq n}(G)$}

We first consider the $k = n$ case.

\begin{theorem}
Given a bipartite graph $G = (A,B,E)$ with $|A|=|B|=n$, it is $\NP$-hard to compute $\a_{\leq n}(H)$.
\end{theorem}
\begin{proof}
It is immediate that $\a_{\leq n}(H) \geq 1/n$ and that equality holds if and only if $H$ contains a connected matching of size $n$. The theorem follows from \cref{thm:hardness-perfect-connected-matching}.
\end{proof}

We proceed and consider approximating $\a_{\leq n}(G)$.
\begin{theorem}
Unless $\NP = \coRP$, there is no polynomial algorithm for approximating $\a_{\leq n}(H)$ within some constant factor.
\end{theorem}
\begin{proof}
We first use the fact that it is \NP-hard to distinguish between $n$-vertex graphs with cliques of size $b \cdot n$ to graphs with no clique of size $a \cdot n$ where $a,b$ are some constants satisfying $1/2<a<b<1$. Indeed it is well known that there are $a,b \in (0,1)$ such that it is \NP-hard to distinguish between $n$-vertex graphs with cliques of size $b \cdot n$ and graphs with no clique of size $a \cdot n$ (e.g.~\cite{haastad2001some}). The fact now follows by taking a graph $G$ of $n$ vertices, adding to it a clique of size $n$ and connecting all vertices in this clique to all vertices of $G$.

Given a graph $G$ apply the reduction in \cref{claim:reduction-cm} (with $H$ being the random graph described in \cref{claim:random-half-cover-small-nu-c}) and call the resulting graph $G'$.
If there is a clique $G$ of size $b \cdot n$ then clearly $\a_{\leq n}(G') \leq \frac{b}{n}$. Suppose there is no clique of size $a \cdot n$ in $G$. Then by \cref{claim:random-half-cover-small-nu-c}, with high probability there is no connected matching in $G'$ of size greater than $(a+\delta )\cdot n$ where $\delta>0$ can be taken to
be arbitrarily small. It follows that for $c>a+\delta$, every connected matching in $G$ contains a induced matching of size at least $2$. Therefore, for $(a + \delta)<c<1$
we have that conditioned on the existence of a matching of size $k$, $\a_{k}(G')=\frac{2}{c n}> \frac{1}{(a+\delta) \cdot n}$. Indeed, $\frac{2}{c}>\frac{1}{a+\delta}$ as $a+\delta>1/2$. As for $k<(a+\delta) n$ it clearly holds that $\a_{k}(G')>\frac{1}{(a+\delta)n},$ we have that in this case
$\a_{\leq n}(G') = \frac{a+\delta}{n}$. This implies that approximating $\a_{\leq n}(H)$ within a ratio smaller than $\frac{b}{a+\delta}$ in polynomial time would allow one to determine whether $G$ contains a clique of size $b \cdot n$ or no clique of size $a \cdot n$. Taking $\delta$ such that $\frac{b}{a+\delta}>1$ concludes the proof.

\end{proof}

\subsection{Hardness results for computing $\a_{\leq k}(G)$ for $k < n$}

We now turn to the problem of proving hardness of approximation results for $\a_{\leq k}(G)$ for $k<n$; for certain values of $k$, we show that $\a_{\leq k}(G)$ is $\NP$-hard to approximate to within any constant factor under randomized reduction. One approach to prove this is to use the reduction in \cref{thm:alpha-n-hardness}. However, this approach does not seem to work, as it allows one to consider also matchings that contain ``diagonal edges" of the form $(u_i,v'_j)$ and it is not clear how to apply the analysis in \cref{thm:alpha-n-hardness} to such matchings.
Instead, we build upon the hardness of the connected matching problem given in \cref{thm:connected-match-inapprox}.
We claim that the reduction in \cref{thm:connected-match-inapprox} shows that it is hard to approximate $\a_{\leq k}(G)$ for $k= n^{1-\eps}$.
Note that in the YES-case, if $\nu_c(G) = k \geq n^{1-\eps}$, then $\a_{\leq k}(G) = 1/k$.
The NO-case is a bit subtle, and it is, a priori, not clear why $\nu_c(G) \leq n^{\eps}$ implies that any matching of size at most $k$ contains a large induced matching. We resolve this problem using the following Ramsey-theoretic fact (see e.g.,~\cite{boppana1992approximating,erdos1935combinatorial}).

\begin{fact}\label{fact:ramsey}
Let $G$ be an $n$-vertex graph not containing a clique of size $k+1$ and suppose $k \geq 2\log n$. Then $G$ contains an independent set of size at least $s=\log n/\log (k/\log n)$.
\end{fact}

Coupled with \cref{thm:alpha-n-hardness} we prove the following result.

\begin{theorem}\label{thm:alpha-k-constant-hardness}
For any constants $\eps \in (0,1/2)$ and $\rho > 1$, it is $\NP$-hard (under randomized reduction) to approximate $\a_{\leq k}(G)$ within a factor of $\rho$
on bipartite graphs with $n$ vertices on each side for $k = n^{1-\eps}$.
\end{theorem}

\begin{proof}
By \cref{thm:connected-match-inapprox} given a bipartite graph $G$ it is $\NP$-hard to distinguish between the case where $\nu_c(G) \geq n^{1-\eps}$, and the case where $\nu_c(G)  \leq n^{\delta}$ for $\delta = 1/(2\rho)$.

For the YES-case if $\nu_c(G) \geq n^{1-\eps}$, then clearly $\a_{\leq k}(G) = 1/k$ for $k = n^{1-\eps}$.

In the NO-case suppose that $\nu_c(G) \leq n^\delta$, and consider an arbitrary matching $M$ of size $s$ with  $s \leq k$. If $s< 2\delta k$ then clearly $M$ contains an induced matching of size at least $s/(2\delta k)$. Otherwise, contract all edges in $M$.
Denote by $H(M)$ the subgraph induced by the $s$ contracted nodes.
Observe that a subset of nodes in $H(M)$ forms a clique if and only if their corresponding edges in $G$ form a connected matching.  Otherwise,
by the assumption that $\nu_c(G) \leq n^\delta$ we get that $H(M)$ contains no clique of size $n^{\delta}$.
Hence, by \cref{fact:ramsey} we conclude that $H(M)$ contains an independent set of size
at least $\frac{\log s}{\log(n^\delta/\log s)} \geq \frac{1}{2\delta}$ (assuming $n$ is sufficiently large).

Therefore, given a bipartite graph $G$ with $n$ vertices on each side, and $k = n^{1-\eps}$
it is $\NP$-hard to distinguish between the YES-case of $\a_{ \leq k}(G) \leq 1/k$,
and the NO-case of $\a_{\leq k}(G) \geq 1/(2\delta k) = \rho/k$. This concludes the proof.
\end{proof}

We can achieve stronger hardness results under stronger assumptions than $\NP$-hardness. Recall that
the Exponential Time Hypothesis (ETH) postulates that no algorithm of running time $2^{o(n)}$ can decide whether an $n$-variable SAT formula has a satisfying assignment. Assuming ETH we have the following hardness result:

\begin{theorem} \label{thm:almost-poly}
Assuming ETH there exists a $k$ such that given $H = (A,B,E)$ with $|A|=|B|=n$ there is no polynomial time algorithm that approximates $\a_{\leq k}(H)$ within
a factor of $n^{(1/\log \log n)^c}$ where $c>0$ is a universal constant independent of $n$.
\end{theorem}

We will rely on the following simple lower bound on independent sets in graphs of average degree $d_{avg}$ due to Turan.

\begin{lemma}\label{lem:turan_independent_set}
Every $n$-vertex graph with average degree $d_{avg}$ contains an independent set of size at least $\frac{n}{d_{avg}+1}$.
\end{lemma}

\begin{proof}[Proof of Theorem~\ref{thm:almost-poly}]
It is known \cite{manurangsi2017almost} that assuming ETH for $k=n^{1-1/\mathrm{polyloglog}(n)}$ there is no polynomial algorithm that distinguishes between the
case where $H$ contains a bipartite clique with $t$ vertices on each side (YES-case) to the case where every subgraph contained in $H$ with $k' \leq k$ vertices satisfies $|E(H)| \leq {k' \choose 2}/n^{(1/\log \log n)^c}$ (NO-case).
In the first case $\a_{\leq k}(H)=1/k$. In the second case, given a matching $M$ with $|M|=k;$ and $k' \leq k$ we claim that $M$ contains an induced matching of size $\Omega(\max((k'n^{-(1/\log \log n)^c},1))$. The claim is trivially true if $k' \leq n^{(1/\log \log n)^c}$ hence assume $k' > n^{(1/\log \log n)^c}$. Let $H(M)$ be the graph induced on $M$ and let $H'(M)$ be the graph obtained after all edges in $M$ are contracted. Clearly the average degree of $H'(M)$ is $O(k'n^{-(1/\log \log n)^c})$ (see Lemma 2.1 in \cite{alon2017graph}) hence by \cref{lem:turan_independent_set} it contains an independent set $I'$ of size $\Omega(n^{(1/\log \log n)^c})$. It is easily verified that this independent set corresponds to an induced matching contained in $M$ whose size is $\Omega(n^{(1/\log \log n)^c})$. Therefore every matching of size at most $k' \leq k$ contains an induced matching of size $\Omega(\lceil k'n^{-(1/\log \log n)^c})\rceil)$
which implies that $\a_{\leq k}(H)=\Omega (n^{(1/\log \log n)^c}/k)$. It follows that if we could approximate $\a_{\leq k}(H)$ within a factor better than $\Omega(n^{(1/\log \log n)^c}))$ in polynomial time then we could distinguish between the YES and NO cases described above. This concludes the proof.
\end{proof}

\section{Improved construction of multitaskers}\label{sec:construction}

In this section we prove the following theorem.

\begin{theorem}\label{thm:alpha->1/2}
Let $d \leq n$ be positive integers such that $n$ is sufficiently large, and let $\eps \in (0,1)$ be such that $\eps \geq \frac{20 \log d}{\log n}$.
Then, there is a bipartite graph $G$ with $n$ vertices on each side and average degree at least $d/2$, such that
$\a_{\leq k}(G)\geq 1/2-\eps$ for $k = (\frac{1}{101e^5})^{4/\epsilon} \cdot \frac{n}{d^{1+8/\epsilon}}=\frac{n}{d^{1+O(1/\epsilon)}}$.
\end{theorem}

For the proof of \cref{thm:alpha->1/2} we need the following lemma.
We remark that a similar result also appears in \cite{alon2017graph} (proof of Theorem 4.14 in the arXiv version).

\begin{lemma}\label{lemma:high-girth-sparse-graphs-are-good-multitaskers}
Let $G=(A,B,E)$ be a balanced bipartite graph, and let $g$ be the girth of $G$.
Let $t \in \N$ be such that for every subset of vertices $T \seq A \cup B$ satisfying $|T \cap A|= |T \cap B| \leq s \leq t$ it holds that $|E(T)| \leq (2+\beta/g)s$ edges for some $\beta > 0$.
Then $\a_{\leq t}(G)\geq \frac{1}{2} - \frac{1+\beta}{g}$.
\end{lemma}
\begin{proof}
Let $G = (A,B,E)$ with $|A|=|B|=n$ that satisfies the assumptions in the lemma.
and let $M$ be a matching in $G$ of size of size $s \leq t$ .
We show that $M$ contains an induced matching $M'$ of size at least $(\frac{1}{2}-\frac{1+\beta}{g})|M|$.

Let $F$ be the graph whose vertices correspond to the $s$ edges of $M$,
and two vertices in $F$ are connected if the corresponding edges are connected by an edge in $G$.
We show below that $F$ contains an independent set on nearly half of its vertices.
By the assumptions of the claim, the girth of $F$ is at least $g/2$,
and any set of $s$ of its vertices spans at most $(1+\beta/g)s$ edges.
Construct an independent set in $F$ as follows.
As long as $F$ contains a vertex of degree at most $1$ add it to the independent set,
and omit it and its unique neighbor from $F$.
Suppose that this process stops with $h$ vertices.
This implies that the independent set so far has at least $(s-h)/2$ vertices.
If $h=0$ we are done, as the independent set has at least $s/2$ vertices.
Otherwise, in the induced subgraph of $F$ on the remaining $h$ vertices the minimum degree is at least
$2$ and the average degree is at most $2+2\beta/g$. Hence it contains
at most $2 \beta h/g$ vertices of degree at least $3$. Omit these vertices.
The remaining graph is a union of paths and cycles, which may
contain odd cycles, but all cycles in it are of length at least $g/2$.
Therefore this part contains an independent set of size at least $\frac{1}{2} (1-2 \beta /g) \cdot (1 - 2 /g) h$,
which together with the $(s-h)/2$ vertices obtained in the initial process result with an
independent set of size at least
\[
\frac{s-h}{2}+ \frac{1}{2} (1-2 \beta/g) \cdot (1 - 2 /g)h
>
\frac{s-h}{2}+ \frac{1}{2} (1- 2 \beta/g - 2 /g)h
>
\frac{s}{2}-\frac{1+\beta}{g}h
\geq
(\frac{1}{2}-\frac{1+\beta}{g})s
,
\]
as required.
\end{proof}

We can now prove \cref{thm:alpha->1/2}.
\begin{proof}
We start a random bipartite graph $G'$ with $n$ nodes on each side, in which each edge is included independently with probability $p=d/n$.
The following two claims prove the properties required in order to apply \cref{lemma:high-girth-sparse-graphs-are-good-multitaskers}.

\begin{claim}\label{claim:Gnp-short-cycles}
Let $g$ be an even integer such that $2/\eps \leq g \leq 4/\eps$.
Then, with probability $1-\frac{2}{n^{0.3}} \geq 0.99$ the number of cycles of length at most $g$ is upper bounded by $\sqrt n$.
\end{claim}
\begin{proof}
The expected number of cycles of length up to $g$ is upper bounded by
\[
\sum_{s=2}^{g/2}{n\choose s}^2(s!)^2p^{2s} \leq \sum_{s=2}^{g/2}(np)^{2s} \leq \sum_{s=2}^{2/\epsilon}d^{2s} \leq 2d^{4/\epsilon}.
\]
In particular, for $\eps \geq \frac{20 \log d}{\log n}$
the expected number of cycles of length up to $g$ is at most $2d^{4/\epsilon} \leq 2 n^{1/5}$.
The claim follows by Markov's inequality.
\end{proof}

\begin{claim}\label{claim:G(n,p)-low-avg-degree}
With probability $0.99$, every subgraph of $G'$ with at most $(\frac{1}{101e^5})^{4/\epsilon}\cdot n/d^{1+8/\epsilon}$ nodes on each side has average degree at most $(2+\epsilon/4)$.
\end{claim}
\begin{proof}
Let $s$ be an integer satisfying $1 \leq s \leq (\frac{1}{101e^5})^{4/\epsilon}\cdot n/d^{1+8/\epsilon}$.
By the union bound over all subsets of $G'$ with $s$ vertices on each side, the probability that $G'$ contains a balanced subgraph with $s$ nodes on each side and average degree at least $(2+\epsilon/4)$ is
\[
{n \choose s}^2{s^2 \choose (2+\epsilon/4)s}p^{(2+\epsilon/4)s}
\leq
\left(\frac{n e}{s}\right)^{2s} \cdot \left( s e \right)^{(2+\eps/4)s} \cdot \left( \frac{d}{n}\right)^{(2+\eps/4)s}
\leq
\left(\frac{e^5 d^{2+\epsilon/4}s^{\epsilon/4}}{n^{\epsilon/4}}\right)^s
\leq
\left(\frac1{101}\right)^s.
\]
By taking the union bound,over all values of $s$ we get that the probability that
$G'$ contains a dense induced subgraph is at most $\sum_{s=1}^\infty\left(\frac1{101}\right)^s=0.01$, as required.
\end{proof}

By Chernoff bound with probability $0.99$, $G'$ contains at least $0.9 dn$ edges.
Therefore, with probability $0.97$ the latter event occurs together with the events in the two foregoing lemmas.

Let $g \in [\frac{2}{\eps},\frac{4}{\eps}]$ be an even integer, as in \cref{claim:Gnp-short-cycles}.
We remove an edge from each cycle of length at most $g$, thus removing at most $\sqrt{n}$ edges, so that the average degree remains at least $d/2$.
The resulting graph $G$ satisfies the conditions of \cref{lemma:high-girth-sparse-graphs-are-good-multitaskers} with $g \in [\frac{2}{\eps},\frac{4}{\eps}]$
and $t = (\frac{1}{101e^5})^{4/\epsilon}\cdot n/d^{1+8/\epsilon}$,
and hence $\a_t(G) \geq 1/2 - 2/g \leq 1/2 - \eps$, as required. This concludes the proof of \cref{thm:alpha->1/2}.
\end{proof}

\begin{remark}
We note that if we consider $\a_{\leq n}(G)$ instead of $\a_{\leq n/d^{1+O(1/\epsilon)}}(G)$,
then for the construction in the proof of \cref{thm:alpha->1/2}
it holds that $\a_{\leq n}(G)=O(\frac{\ln d}{d}+O(1/\sqrt{n}))$ with high probability.
Indeed, it can be shown that prior to deletions $G'$ has a matching of size $\Omega(n)$ and no induced matching of size larger that $O(\frac{\ln d}{d}n)$ with high probability. Therefore, since removing $\sqrt{n}$ edges can increase the size of any induced matching by at most $\sqrt{n}$,
we get that the entire construction satisfies $\a_{\leq n}(G)=O(\frac{\ln d}{d}+1/\sqrt{n})$.
We also remark that \cite{alon2017graph} described a construction with average degree $d=\Omega(\log \log n)$ for which $\a_{\leq n} > 0.33$.
\end{remark}

\subsection{Is $\alpha_k(G)=1/2$ attainable?}

The foregoing positive result obtains $\alpha_{\leq k}(G)=1/2-\epsilon$ for $k=O(n/d^{1+O(1/\epsilon)})$, approaching the natural barrier $1/2$.
A natural question is whether $1/2$ can be attained exactly, and for which values of $k$.
We now show the following limitation.

\begin{theorem}
There are absolute constants $C>0$ and $\epsilon_0>0$ such that for $n,d$ sufficiently large and $k\geq C\cdot n/d^{1+\epsilon_0}$, every graph $G$ with $n$ nodes on each side and average degree $d$ has $\alpha_{\leq k}(G)$ strictly smaller than $1/2$.
\end{theorem}
\begin{proof}
One obstacle to obtaining $\alpha_{\leq k}(G)=1/2$ is cycles of length $2$ mod $4$.
In particular, consider a cycle of length $\ell=2k$, where $k$ is an odd integer.
It is straightforward to check that picking every other edge of the cycle yields a matching $M$ of size $k$, whose largest induced matching has size $(k-1)/2=(\frac12-\frac1{2k})|M|$. Hence, a graph $G$ containing such cycle has $\alpha_k(G)$ strictly less than $1/2$.
We now show that every $n$-vertex graph with average degree $d$ contains such a cycle for some $k = O(n/d^{1+\epsilon_0})$.

Let $G$ be a bipartite graph with $n$ nodes on each side and average degree $d$.
It is known (see e.g.,~\cite{sudakov2016extremal}) that there is an absolute constant $r>0$ such that any graph with average degree at least $r$ contains a cycle of length $2$ mod $4$, which we call a ``bad'' cycle.
By~\cite{feige2016generalized}, $G$ contains a subgraph $G'$ on at most $\ell=O(n/d^{1+2/(r-1)})$ nodes whose average degree is at least $r$. Hence $G'$ (and therefore $G$) contains a bad cycle whose length is at most $\ell$.
As per above, this implies that $\alpha_{\leq k}(G)$ is strictly less than $1/2$ for $k=\ell/2$.
\end{proof}

We do not know if $\alpha_{\leq k}(G)=1/2$ is achievable for the value $k$ considered in \cref{thm:alpha->1/2}.
We suspect that $\alpha_{\leq k}(G)<1/2$ holds for all for graphs of large enough average degree $d$
even for $k = \poly\log n$. Whether this is indeed the case is left as an open question.
\section{Conclusion and future directions}
We have studied the computational complexity of computing $\a_k(G)$, a parameter that arises in wireless networks and
connectionist models of multitasking. Our study reveals that algorithmic as well as combinatorial questions
(such as the existence of graphs with certain combinatorial properties) are relevant to connectionist models of cognitions,
and we hope that future work will reveal more connections between such models, theoretical computer science and combinatorics.

While we have shown that computing $\a_k(G)$ is intractable, our results do not rule out the existence of an
efficient constant factor approximation algorithm for $\alpha_{\leq n}(G)$, which could potentially be used in
in computer simulations and in analyzing behavioral and neuroscientific data. Whether such an algorithm exists
is an interesting direction for future study.

We conclude with several specific questions arising from this work.

\begin{itemize}
\item We believe that for $d$-regular graphs the upper bound $\alpha_{\leq n}(G) \leq 9/ \sqrt{d}$ is not tight.
It is an open problem whether for all $d$-regular graphs it holds that $\alpha_{\leq n}(G) \leq o(1/\sqrt{d})$,
and it is possible that $\alpha_{\leq n}(G)=O(\frac{\log d}{d})$ holds.
\item It would be interesting to see if the $n^{1-\epsilon}$ hardness of approximation result can be obtained assuming $\P \neq \NP$
(that is, under a deterministic reduction).
In particular, it would be interesting to find efficient and deterministic constructions of bipartite half-covers with maximal connected matching upper bounded by $n^{o(1)}$.
\item Finally, understanding how well one can multitask on a ``small" number of tasks is of interest. This raises the question of fixed-parameter algorithms for $\a_k$ and connected matchings, where $k$ is a parameter independent of $n$.
\end{itemize}

\bibliographystyle{alpha}
\bibliography{approx}
\appendix

\section{$\NP$-hardness of computing the maximum connected matching of a graph}
\label{sec:hardness-deterministic-reduction-cm}

In this section we that given a bipartite graph it is NP-hard to compute $\nu_C(G)$ exactly under a deterministic polynomial time reduction.
This is as opposed to the randomized reduction given in \cref{thm:connected-match-inapprox}.
We remark that \cite{plummer2003special} proved this result for the non-bipartite case.
Our proof is an adaptation of their proof to the bipartite case.

\begin{theorem}\label{thm:exact_connected}
It is \NP-hard to determine given a bipartite graph $G=(A,B)$ and a parameter $k$ whether $G$ contains a connected matching of size $k$.
\end{theorem}

\begin{proof}
We reduce the biclique to the problem of determining if $\nu_C(G)=k$. Recall that a biclique $G'=(C',D')$ in a bipartite graph $G$ is a subgraph $G'$ of $G$ such that every vertex in $C'$ is connected to every vertex in $D'$. A biclique $(C',D')$ is balanced if $|C'|=|D'|$ The biclique problem is: given a bipartite graph $G = (A,B)$ (we assume that $|A| = |B|$) and integer $k$, is there a biclique $(A',B')$ with $A' \subseteq A, B'\subseteq B$ and $|A'| = |B'| = k$. This problem is well known to be \NP-complete.

Given a bipartite graph $G = (A,B)$ with $|A| = |B| = n$, form a new graph $H$ as follows. Initialize $H_1 = (A_1,B_1)$ to equal $G$ and we call this the copy of $G$ inside $H_1$. Then add a new set $A'$ of $n$ vertices such that $(A_1,A')$ forms biclique, and add a new set $B'$ of $n$ vertices such that $(B_1,B')$ forms a biclique. Initialize another graph $H_2 = (A_2,B_2)$ to be a biclique with $|A_2| = |B_2| = n$ (where $A_2,B_2$ are disjoint from $A_1\cup B_1 \cup A' \cup B'$). Add an edge between every vertex of $(A_1 \cup B')$ and every vertex of $B_2$, and add an edge between every vertex of $(B_1 \cup A')$ and every vertex of $A_2$. The resulting (bipartite) graph is $H = (A_1 \cup B' \cup A_2, B_1 \cup A' \cup B_2)$.

Consider a connected matching $M$ in $H$. Let $M_A \subseteq M$ be the set of all edges in $M$ contained in the biclique $(A_1,A')$ and let $M_B \subseteq M$ be the set all edges in $M$ contained in the biclique $(B_1,B')$, and let $M_r = M - (M_A \cup M_B)$. Then $|M| = |M_A| + |M_B| + |M_r|$. Let $X_A\subseteq A_1$ denote the set of vertices in $A_1$ being an endpoint of an edge in $M_A$, and let $X_B$ be analogously defined with respect to $B_1$ and $M_B$. Since $M$ is a connected matching, $(X_A, X_B)$ is a biclique. We also have $|M_r| \leq 2n - \max\{|X_A|,|X_B|\}$ which implies $|M| \leq 2n + \min\{|X_A|,|X_B|\}$ where we have used $|X_A| = |M_A|, |X_B| = |M_B|$. Thus, if $G$ has a connected matching of size $2n + k$ then $\min\{|X_A|,|X_B|\} \geq k$ which means that there is a biclique of size $k$.

Conversely, if $G$ contains a biclique $(R,S)$ of size $k$, we can easily form a connected matching $M$ in $H$ of size $2n + k$. To construct $M$, we take $k$ edges $M_A$ in $(A,A')$ with $X_A = R$, we take $k$ edges $M_B$ in $(B,B')$ with $X_B = S$, we take $n - k$ edges matching the $n - k$ vertices of $A_1 - X_A$ with $n - k$ vertices $B_2' \subseteq B_2$, we take $n - k$ edges matching the $n - k$ vertices of $B_1 - X_B$ with $n - k$ vertices $A_2' \subseteq A_2$, and we take $k$ edges matching $A_2 - A_2'$ with $B_2 - B_2'$.

Thus, $G$ contains a biclique of size $k$ if and only if $H$ contains a connected matching of size $2n + k$. This completes the proof.
\end{proof}
\end{document}